\documentclass[11pt]{article}
\usepackage{subfig}
\usepackage{amssymb, graphicx}
\usepackage{graphics}
\usepackage{amsmath}
\usepackage{mathrsfs}
\usepackage[round]{natbib}
\usepackage{fullpage}
\usepackage{sidecap}
\usepackage[pdftex]{hyperref}

\newtheorem{theorem}{Theorem}
\newtheorem{lemma}[theorem]{Lemma}
\newtheorem{proposition}[theorem]{Proposition}
\newtheorem{corollary}[theorem]{Corollary}
\newtheorem{example}{Example}

\newtheorem{definition}{Definition}

\newtheorem{claim}{Claim}

\newcommand{\ignore}[1]{}

\newcommand{\ga}{\gamma}

\newcommand{\polylog}{{\rm polylog}}

\def\squarebox#1{\hbox to #1{\hfill\vbox to #1{\vfill}}}
\newcommand{\qed}{\hspace*{\fill}
\vbox{\hrule\hbox{\vrule\squarebox{.667em}\vrule}\hrule}\smallskip}

\newenvironment{proof}{\noindent{\bf Proof:~~}}{\(\qed\)}

\date{October 2010}

\begin{document}

\title{Matching with Couples Revisited}

\author{Itai Ashlagi \and Mark Braverman \and Avinatan Hassidim\thanks{Ashlagi: Sloan School of Management, MIT, iashlagi@mit.edu. Braverman: Department of CS, University of Toronto, mbravem@cs.toronto.edu. Hassidim: Google, Israel, avinatanh@gmail.com. We thank Shahar Dobzinski, Itay Fainmesser, John Hatfield, Jacob Leshno, Brendan Lucier and Scott Kominers for very helpful discussions and comments.}}

\maketitle

\begin{abstract}
It is well known that a stable matching in a many-to-one matching market with couples need not  exist.
We  introduce a new matching algorithm for such markets and show that for a general class of large  random markets the algorithm will find a stable matching with high probability. In particular we allow the number of couples to grow at a near-linear rate. Furthermore, truth-telling is an approximated equilibrium in the game induced by the new matching algorithm.
Our results are tight: for markets in which the number of couples grows at a linear rate, we show that
with constant probability  no stable matching exists.
 \end{abstract}

\section{Introduction}
\label{sectionIntroduction}

We consider a many-to-one matching market, in which one side of the market consists of hospitals and the other consists of doctors.  Stability is the  most natural and desired property in such markets.  Therefore understanding  when a stable matching exists in a matching market with couples as well as providing  an efficient procedure to find one (whenever it exists) are both important tasks, and are the main scope of this paper.

 \citet{GaleShapley} introduced the well-known Deferred Acceptance algorithm and showed that if doctors' preferences do not depend on other doctors' preferences, in other words all doctors are ``single", the algorithm will always produce a stable matching. Naturally, when couples are present in the market,  their preferences  depend on each other and often introduce complementarities,  a stable matching may not exist (\citet{Roth1984}). In fact, for any market size one can construct a preference profile for which a stable matching does not exist and even if a stable matching does exist, finding it can be computationally intractable (\cite{Ronn}).

Several clearinghouses exist today for two sided markets with couples. Two major examples are the National Resident Matching Program (NRMP)  and the
clearinghouse for psychology interns. Until not long ago  couples had to participate as singles,  since  clearinghouses for these markets used the
Deferred Acceptance algorithm to find a matching.  Only since 1999, the NRMP  and the psychology market  adopted the  new algorithm  designed
by  \citet{RothPeranson}  allowing for couples to express their joint preferences, henceforth called the Roth-Peranson (RP) algorithm.
 This algorithm has had a great success in practice: every year since it is used, the NRMP has found a stable matching with respect to
 the reported preferences. For a comprehensive background, and history of these markets see \citet{RothCouples,RothCouplesHistory}.


 \citet{KlausCouples1}\footnote{See also \citet{KlausCouples2}.} initiated the study of characterizing  markets with couples  that have a stable matching. They showed that the  domain of {\it responsive preferences} is a maximal domain in which a stable matching exists. However,  \citet{RothCouples} observe from real data that  couples' preferences often do not belong to this domain. Adopting a random preferences approach, they used a much simplified version of RP (that attempts to find a stable matching) to show that if there are $n$ single doctors,
 and the number of couples is of order $\sqrt{n}$ a stable matching exists with probability converging to one as $n$ approaches  infinity.
 The approach for studying random growing markets is well founded.\footnote{\citet{NicoleMahdian} and \citet{KojimaParag} also used a similar large market approach to study incentives and stability  in a one-to-one and many-to-one matching markets without couples.} About  16,000 single doctors and  800 couples participated in the NRMP in 2010\footnote{In fact there were about 40,000 doctors, but only 16,000 of from American institutions and most couples were from American institutions.}, and about 3,000 single doctors and 19 couples participated in the psychology clearinghouse in the same year. Furthermore these figures are increasing every year. While the size of the market justifies the large-market-assumption, the number of couples increases every year. Thus  although  results by \citet{RothCouples} can explain the success in the psychology  market,  the success of the NRMP market remains a puzzle as the the number of couples is much larger than $\sqrt{n}$.

We  introduce  a new matching algorithm, called {\it Sorted Deferred Acceptance} (SoDA),
for many-to-one matching markets with couples. Our approach is slightly different from RP, although the algorithms share similar ideas.
The  SoDA algorithm is simple and consists of two main steps; First
(i)  it finds a  stable matching in the sub-market without couples using Deferred Acceptance. Then (ii) in some given {\it order}, each couple $c$
applies according to its preference list; whenever a single is rejected it applies until it finds a position.
If  some other couple $c'$  has been  rejected after being assigned, the second step starts over,  letting however $c'$ apply just ahead of $c$.

As noted above we study large markets  and analyze the  performance of the SoDA algorithm in these markets.   In our model, hospitals have capacities, there is  an excess number of available positions\footnote{There are $\lambda n$ positions for some $\lambda>1$.}, all doctors are acceptable to all hospitals and vice versa,  doctors preferences are random and hospital preferences are arbitrary. As we will show, when doctors' preference lists are long, an excess number of positions is  necessary for the existence of stability even if there is only one couple in the market.\footnote{\citet{RothCouples} do not assume  an excess number of positions.
They assume, however, that doctors have `short' preference lists, and show that it results in an excess number of positions.}
All our main results hold without restricting the length of doctors' preference lists.

We first provide positive results for a near-linear rate.
If the number of couples grows at a rate of at most $n^{1-\epsilon(n)}$ where $\epsilon(n)$ is a `slowly' decreasing function converging  to zero:\footnote{$\epsilon(n)$ can be replaced by any fixed $\epsilon>0$.}
\begin{enumerate}
\item[1.] The probability that a stable matching exists and is found by the SoDA algorithm approaches 1 (as $n$ approaches infinity).

\item[2.]  The probability that any doctor or any couple can gain by misreporting her preferences converges to 0 even ex post. A similar result can be shown for hospitals, implying that truth-telling is an approximated Bayes Nash equilibrium in the game induced by SoDA for any large enough $n$.
\end{enumerate}
Note that if $\epsilon(n)$ is approximately $1/\log n$ then the growth rate of couples is linear. Our result holds for any $\epsilon(n)=\Omega(\log\log n/\sqrt{\log n})$
(see the last section for further discussion).\footnote{We write $f(n)=\Omega(g(n))$ if there exists $c>0$ such that $f(n)\geq cg(n)$ for every large enough $n$.}


Our first result is tight in the following sense. When the number of couples grows at a rate of $\alpha n$ for some $\alpha>0$ we show:
\begin{enumerate}
\item[3.]
For some $\lambda>2\alpha+1$, if the number of hospitals is $\lambda n$, with constant probability (not depending on n) no stable matching exists (even if hospitals' preferences are random).

\end{enumerate}
While the third result does not cover the case when the excess is small, we believe that the result remains true in this case.
We  give evidence, based on simulation, that in the setting where the number of couples is linear, the probability of failure decreases as $\alpha$ decreases but remains constant as $n$ grows.
One can thus view our results as a characterization for the existence of a stable matching with high probability in a large random market with couples.

We  also show that  the SoDA algorithm runs in polynomial time (in fact `almost' linear),  and  provide simulations that test SoDA in various large random markets. Finally, we believe our proof technique is interesting for its own sake, and may serve as a tool for in the search for positive results in  other settings  with complementarities. Some open problems are discussed in the last section.

SoDA is the first algorithm for matching markets with couples that is proven to find a stable outcome under very general settings.
The provable success of the SoDA algorithm
helps explain the fact that algorithms, RP in particular, have been successful in finding stable matchings in real life.\footnote{We believe that our techniques can be adapted to prove that directly that the RP algorithm
also succeeds with high probability in large random markets.}
This adds to the short list of positive results in settings with complementarities (see e.g.  \citet{Milgrom04},  \citet{Gul99}, \citet{SunYang} and \citet{ParkesLahaie} for auction settings, and \citet{Kominers1} and \citet{Pycia1} for matching settings).

\section{Matching Markets with Couples}
\subsection{Model}
\label{model}

In a matching market there is a set of hospitals $H$ a set of single doctors $S$ and a set of couples of doctors $C$. Each single doctor $s\in S$ has a
strict  preference relation $\succ_s$ over the set of hospitals.
Each couple $c\in C$ denoted by $c=(f,m)$  has a strict  preference relation
$\succ_c$ over pairs of hospitals. For every couple $c$ we denote by $f_c$ and $m_c$ the first and second  members of $c$. Denote by $D$ the set of all doctors. That is  $D=S\cup \{m_c| c\in C \} \cup \{f_c| c\in C\} $.
Each hospital $h\in H$ has a fixed capacity $k_h>0$ and a strict preference relation $\succ_h$ over the set $D$. For any set $D'\subseteq D$ hospital $h$'s {\bf choice} given $D'$, i.e. the most preferred doctors $h$ can employ, $CH_h(D')$, is induced by $\succ_h$ and $k_h$ as follows: $d\in D'\cap CH_h(D')$ if and only if there exist no set of $D''\subseteq D'\setminus \{d\}$ such that $|D''|=k_h$ and $d'\succ_h d$ for all $d'\in D''$.

A {\bf matching} $\mu$ is a function from $H\cup C\cup S$ such that  $\mu(s)\in H\cup \{\phi\}$ for every $s\in S$,
$\mu(c)\in H\times H\cup \{(\phi,\phi)\}$ for every $c\in C$,  $\mu(h)\in 2^D$ for every $h\in H$, and:
\begin{enumerate}
\item[(i)]  $s\in \mu(h)$  if and only if $\mu(s)=h$.
\item[(ii)]   $\mu(c)=(h,h')$ if and  only if $f_c\in\mu(h)$ and $m_c\in\mu(h')$.
\end{enumerate}
$\mu(s)=\phi$ means that $s$ is unassigned under $\mu$, and similarly $\mu(c)=(\phi,\phi)$ means that the couple $c$ is unassigned under $\mu$.





We proceed to define stability.
Blocking coalitions for a given matching can be formed in several ways:
\begin{itemize}
\item  $(s,h)\in S\times H$ is a {\bf block}  of $\mu$ if $h\succ_s\mu(s)$ and $s\in CH_h(\mu(h)\cup s)$.
\item  $(c,h,h')\in C\times H\times H$ (where  $h\neq h'$) is a {\bf block} of $\mu$ if  $(h,h')\succ_c\mu(c)$,   $f_c\in Ch_h(\mu(h)\cup f_c)$, and  $m_c\in Ch_{h'}(\mu(h')\cup m_c)$
    \item  $(c,h)\in C\times H$ is a {\bf block} of $\mu$ if $(h,h)\succ_c\mu(c)$ and $\{f_c,m_c\}\in Ch_h(\mu(h)\cup c)$.
\end{itemize}
Finally a matching is {\bf stable} if there is no block of $\mu$.

Gale and Shapley (1962)  showed that the (doctor proposing) Deferred Acceptance algorithm described below,  always produces a stable matching in a matching market without couples.  They further showed that the stable matching produced by this algorithm is the one that is weakly preferred by all single doctors. Roth~(1982) showed that the mechanism induced by this algorithm makes it a dominant strategy for all single doctors to report their true preferences.

\noindent{\bf Doctor-Proposing Deferred Acceptance Algorithm (DA):}
\paragraph{}

{\it \noindent\emph{Input:} a matching market $(H, S, \succ_H, \succ_S)$ without couples.
\paragraph{}

\noindent \emph{Step 1:} Each single doctor$s\in S$ applies to her most preferred hospital. Each hospital rejects its
least preferred doctor in excess of its capacity
among those who applied to it, keeping the rest of the doctors temporarily.

\paragraph{}
\noindent \emph{Step t:} Each doctor who was rejected in Step t-1 applies to her next highest  choice if such exists.
Each hospital considers these doctors as well as the doctors who are temporarily held from the
previous step, and rejects the least-preferred doctors in excess of its capacity
keeping the rest of the doctors temporarily.

\paragraph{}
\noindent The algorithm terminates at a step where no doctor is rejected.}

\paragraph{}

In the next section we introduce a new algorithm for finding a matching in a market with couples. Roth~(1984) showed that when there are couples, sometimes a stable match does not exist. In Section \ref{sec:stability} we show that this algorithm produces a stable matching with very high probability when there is a large market with the number of couples growing (almost) linearly.

\subsection{A New Matching Algorithm}

The matching algorithm  that we present here first finds the stable matching in the market without couples (using DA)
and then attempts to insert the couples, while maintaining the deferred acceptance idea.

Informally, the new algorithm receives as input a matching market with couples and does the following:

{\it
\begin{enumerate}
\item[\emph{(i)}] Find the stable matching in the sub-market without couples using the DA algorithm.
\item[\emph{(ii)}] Fix an order $\pi$ over the couples. Let each couple $c$ on its turn according to $\pi$ apply to pairs of hospitals according to its preference list $\succ_c$ (beginning with the most preferred) and once it finds a pair of hospitals that accepts it, we assign the couple to the pair of hospitals and {\emph stabilize} the current matching as follows:
      \begin{enumerate}
      \item[\emph{Stabilize:}] Continue the DA algorithm, with the singles that were rejected from the their positions in the pair of hospitals that the last couple $c$ was assigned to (at most two singles).

          If during stabilizing one of the members of the last couple $c$ was rejected the algorithm fails. Otherwise if some other couple $c'\neq c$ was rejected during stabilizing, the order $\pi$ is changed so that $c$ is moved one place ahead of $c'$ and part (ii) begins again with the altered permutation;  If the new order $\pi'$ has been tried previously the algorithm fails.
      \end{enumerate}
\end{enumerate}}

Note that if the algorithm terminates without failure it produces a stable matching. As mentioned  in the previous section, this algorithm will serve as a main tool in showing that there exist a stable matching in a large random market.

\citet{RothCouples} used  a similar algorithm  but allows couples to apply in one order, i.e. if some couple is evicted their algorithm fails, even though {\it there might  be a different order  of couples' applications which will not lead to such a failure.}
In our  algorithm if some couple has been rejected the algorithm allows couples to a apply again using a different ordering.

We next describe our algorithm formally.

\noindent{\bf Sorted Deferred Acceptance Algorithm (SoDA):}

{\it
\noindent \emph{Input:} A matching market $(H, S, C, \succ_S, \succ_H, \succ_C)$ and a default permutation $\pi$ over  the set $\{1,2,\ldots, |C|\}$. Let $\Pi=\phi$.

\noindent \emph{Step 1:} Find the stable matching $\mu$  produced by the DA algorithm in the matching market $(H,S,\succ_S,\succ_H)$ without couples.

\noindent  \emph{Step 2 [Iterate through the couples]:} Let $i=1$ and let $B=\phi$.

\begin{enumerate}
\item[(a)] Let $c=c_{\pi(i)}$ be the $\pi(i)$-{th} couple.

Let $c$ apply to the most preferred pair of hospitals $(h,h')\in H\times H$ that has not rejected it yet. If such a pair of hospitals does not exist, modify $\mu$ such that $c=(f,m)$ is unassigned and go to step 2(a) with $i+1$.
    If such a pair $(h,h')$ exists then:

\begin{enumerate}
\item[(a1)] If $h=h'$ and $\{f,m\}\subseteq Ch_h(\mu(h)\cup c)$ then:

        Let $R = \mu(h)\setminus Ch_h(\mu(h)\cup c)$ be the rejected doctors from $h$.
        \begin{enumerate}
        \item[(a11)]
        If there exist a couple $c'\neq c$  for which $\{f_{c'},m_{c'}\}\cap R$ then:
        Let $j<i$ be such that $c_{\pi(j)}=c'$. Let $\pi'$ be the permutation obtained by $\pi$ as follows:

        $\pi'(j)=\pi(i)$, $\pi'(l)=\pi(l)$ for all $l$  such that $l<j$ or $l>i$ and $\pi'(l)>\pi(l-1)$  for other $j+1\leq l\leq i$.

        If $\pi'\in \Pi$ terminate the algorithm. Otherwise add $\pi'$ to $\Pi$ and go to Step 1 setting $\pi=\pi'$.

        \item[(a12)]
        Modify $\mu$ by assigning $c$ to $h$,   remove $R$ from $\mu(h)$. Add $R$ to $B$ and do Step 3 \emph{(Stablize)} with the couple $c$.
        \end{enumerate}

\item[(a2)]  If $h\neq h'$, $f\in Ch_h(\mu(h)\cup f)$, and $m\in Ch_{h'}(\mu(h)\cup m)$ then:

        Let $R_h = \mu(h)\setminus Ch_h(\mu(h)\cup \{f\})$ and $R_{h'} = \mu(h')\setminus Ch_{h'}(\mu(h')\cup \{m\})$.
         \begin{enumerate}
        \item[(a21)]
        If there exist a couple $c'\neq c$ for which $\{f_{c'},m_{c'}\}\cap (R_h\cup R_{h'})$ then:
        Let $j<i$ be such that $c_{\pi(j)}=c'$, change $\pi$ as in step 2(a11). If $\pi\in \Pi$ terminate the algorithm. Otherwise add $\pi$ to $\Pi$ and go to Step 1.
        \item[(a22)]
        Modify $\mu$ by assigning $f$ to $h$ and $m$ to $h'$, remove $R_h$ from $\mu(h)$ and  remove $R_{h'}$ from $\mu(h')$. Add $R_h\cup R_{h'}$ to $B$ and go to Step 3 \emph{(Stablize)} with the couple $c$.
        \end{enumerate}

\item[(a3)] Otherwise, let $h$ and $h'$ reject the couple $c$ and go to Step 2(a).

\end{enumerate}

\end{enumerate}

\noindent \emph{Step 3 [Stabilize]:}  Let $j=|B|$.   As long as $j\geq 0$:

\begin{enumerate}
\item[(a)] If $j=0$ increment  $i$ by one and got to Step 2.
\item[(b)] Otherwise pick some $s\in B$ and:
\begin{enumerate}
\item[(b1)] Let $h$ be the most preferred hospital $s$ has yet to apply to. If such a hospital does not exist then modify the matching $\mu$ such that $s$ is unassigned and go to Step $2(a)$. Otherwise:

Let $R = (\mu(h)\cup \{s\})\setminus Ch_{h}(\mu(h)\cup \{s\})$.

\begin{enumerate}
\item[(b21)]
If  $\{f_{c},m_{c}\}\cap R$ then the algorithm fails.
\item[(b22)]
If there exist a couple $c'\neq c$ for which $\{f_{c}',m_{c}'\}\cap R$ then let $i$ and $j$ be such that $c_{\pi(i)}=c$ ($c$ is the last couple that applied) and $c_{\pi(j)}= c'$.
        Change $\pi$ as in Step 2(a11). If $\pi\in \Pi$ terminate the algorithm. Otherwise add $\pi$ to $\Pi$ and go Step 1.

\item[(b23)]
        If $s\in R$ then go to Step 3(b1).
\item[(b24)]
        Modify $\mu$ by assigning $s$ to $h$, remove $R$ from $\mu(h)$. Add $R$ to $B$ and go to Step 3.
\end{enumerate}
\end{enumerate}
\end{enumerate}
}

Observe that the SoDA algorithm fails to produce a matching in two cases: (i) if  a couple $c$ that finds a pair of positions causes a ``chain reaction" leading to the same couple $c$ being rejected (step 3(b21)),  or (ii) it is about to let couples apply in and order that has already been tried before (steps 2(a11), 2(a21) and 3(b22))  (it changes the permutation  $\pi$ to a permutation $\pi'$ that already belongs to $\Pi$). As mentioned above, if the algorithm does not fail it produces a stable matching.

The following definition will be useful throughout the paper.
\begin{definition}[Evicting]
Let $d\in D$ be a doctor and suppose that $d$ is (temporarily) assigned to some hospital $h$. Let $c\in C$. If during the execution of the  SoDA algorithm some member of the couple $c$ who is not assigned to $h$ applies to $h$ and  causes $d$ to be rejected by $h$, we say that $d$ was {\bf evicted} by $c$. Furthermore, if $d$ was evicted by $c$, applies to some hospital $h'$ and causes some other doctor $d'$ who is assigned to $h'$ to be rejected, we also say that $d'$ is evicted by $c$, and so forth. Finally, if $d$ was evicted  by $c$ and $d$ belongs to a couple $c'$ we say that  $c$ was evicted by $c'$. Formally, all doctors in the set $R$ in steps 2(a1), 2(a2) and 3(b2) are evicted by the applying couple $c$.
\end{definition}

{\bf Remark:} According to this definition $c$ can evict itself. Such a phenomenon may occur since one  member of a given couple can evict the other member of the couple (in the algorithm this happens in  part (b21)).


\section{A Large Market Model}

A {\bf random market} is a tuple $\Gamma=(H,S,C,\succeq_H, Z,Q)$ where $Z=(z_h)_{h\in H}$  and $Q=(q_h)_{q\in H}$ are probability distributions over $H$.

The preference list of each single doctor $d\in S$  is independently drawn as follows:  for each $k=1,\ldots,|H|$ given  $s$'s preference list  up to her  $k$-th most preferred hospital, draw independently according to $Z$ a hospital $h$ until $h$ does not appear in $s$'s $k$ most preferred hospitals and let it be $s$'s $(k+1)$-th most preferred hospital.  The preference list for  each couple $c=(f,m)$ is drawn from the distribution $Q\times Q$.

We will assume that the distributions $Z$ and $Q$ are {\it uniformly bounded}, that is there exist $r\geq 1$ such that
$\frac{q_h}{q_{h'}}\in [\frac{1}{r},r]$ and  ~$\frac{z_h}{z_{h'}}\in [\frac{1}{r},r]$ for every $h,h'\in H$. Define  $\gamma_{max}$ to be the maximum probability  that a hospital is drawn either from $Z$ or from $Q$, that is $\gamma_{max}=\max_{h\in H}\max(q_h,z_h)$.


We will consider a  {\bf sequence of random markets} $\Gamma^1,\Gamma^2,\ldots$ where $\Gamma^n=(H^n,S^n,C^n,\succeq_H^n,Z^n,Q^n)$, i.e. markets with a growing size.

\begin{definition}
A sequence of random markets $\Gamma^1,\Gamma^2,\ldots$ is called {\bf regular}
if there exist $0<\epsilon<1$, $\lambda>1$, $c>0$ and $r\geq 1$ such
that for all  $n$
\begin{enumerate}
\item
$|S^n|=n$ and $|C^n|=O(n^{1-\epsilon})$ ~(the number of couples grows almost linearly).

\item for each hospital  $h\in H^n$, $k_h<c$  ~(bounded capacity).

\item $\sum_{h\in H^n}k_h\geq \lambda n$ ~(excess number of positions).
\end{enumerate}
\label{seq-regular}
\end{definition}
Importantly our results are true even if $\epsilon$ is a `slow' decreasing function of $n$ converging to zero. The exact rate is discussed in the last section.

In their model, \cite{RothCouples} assumed that each doctor's preference list is bounded by a constant,  whereas in our setting all hospitals are acceptable to each doctor. A key step in their proof is to show that  the number of unfilled positions grows linearly in $n$ with high probability.
Instead, we assume an excess number of positions. In fact one can show that under long preference lists, an excess number of positions is necessary for the existence of a stable matching even with one couple:\footnote{This is the only result that we require that doctors' preference lists are long - every hospital is acceptable to every doctor.}

\begin{proposition}
  Consider a  matching market with $n-1$ singles, one couple $c$ and $n$ hospitals each of capacity 1.  Then there exist  preferences for hospitals such that no stable matching exists.
  \label{thm:linear2}
\end{proposition}
The proof follows by a simple embedding of  the (elegant) counter example by \citet{KlausCouples1}; In  particular by letting the preference of each hospital $h$ be  $m_c\succ_h s\succ_h f_c$ for every single $s$.



\section{Stability}
\label{sec:stability}

In this section we show:
\begin{theorem}
Let $\Gamma^1,\Gamma^2,\ldots$ be a regular sequence of random markets.
Then the probability that there exists a stable matching tends to $1$ as $n$ goes to infinity.
\label{thm:stable}
\end{theorem}

To prove  Theorem \ref{thm:stable} we will show that for random preferences the probability that the SoDA algorithm ends without failure converges to 1 as $n$ goes to infinity.
Before we prove the theorem we provide some intuition and a brief outline of the proof.

\subsection{Intuition and Proof Sketch}

The goal is to show that if the number of couples is $m=n^{1-\epsilon}$  (for any $0<\epsilon<1$) then as $n$ approaches infinity the probability of a stable match approaches 1.  To better understand our approach we begin with the intuition for why the result holds for any $\epsilon<\frac{1}{2}$ (essentially we provide the intuition for the result by \citet{RothCouples}). Then we provide intuition for how to obtain the result for $\epsilon<\frac{2}{3}$ and finally for any $\epsilon<1$.

\paragraph{}

\noindent{\bf 1. Number of couples is $n^{\frac{1}{2}-\delta}$:}  Consider the following simplified version of the SoDA algorithm which we call the {\it  direct algorithm}: after finding the stable matching in the market without couples, the couples apply one by one and if some couple evicts another couple the algorithm fails (i.e. it does not attempt to change the permutation over the couples).   Observe that if the algorithm does not fail, it outputs a stable matching.

We will therefore bound the probability that a member of a couple will be evicted from a hospital. We do this iteratively.  When the first couple applies, no other couple will be evicted (since there are no couples in the system). When the second couple $c$ applies, what is the probability that it will evict the first couple?

The second couple $c$ creates a ``chain reaction", which can cause several doctors who were temporarily assigned to  continue applying. To bound the length of this chain consider $f_c$. At some point she is temporarily assigned to a hospital $h$. If this hospital's capacity wasn't full, she did not evict any doctor and therefore also no other couple and we are done. Since there are more positions than doctors, the probability that the hospital has a vacancy is $1-\frac{1}{\lambda}$ (for simplicity we assume here that each hospital has capacity one and the preference distributions are uniform). If the hospital has no vacancy, she evicts a doctor $d_1$ who enters some hospital $h_1$. If $h_1$ has a vacancy, we are done. If $h_1$ is full, a doctor $d_2$ gets kicked out, and looks for a new position. Say $d_2$ is assigned to $h_2$. Again, $h_2$ can have a vacancy, or be full, and this goes onwards. However, since at every step of the chain there is a constant probability for a vacancy, one can show that with probability $1 - 1/n^3$ the number of hospitals $h, h_1, h_2,...$  in the chain is upper bounded by $3 \lambda\log n / (\lambda - 1)$. 

Now, we can estimate the probability that the second couple evicts the first. The second couple kicks out doctors from at most
$6\lambda\log n / (\lambda - 1)$ hospitals. If this list includes the hospitals which admitted the first couple, we could be in trouble. But since preferences are random, the chances that the second couple influences any of these hospitals are upper bounded by
\[2 \cdot \frac{6 \lambda\log n}{(\lambda - 1)n} = \frac{12 \lambda\log n}{(\lambda - 1)n}\]
What about the third couple? Again, it influences at most $6 \lambda\log n / (\lambda - 1)$ hospitals. But now, there are four hospitals which must not be influenced: two hospitals (at most) for each previously assigned couple.
Generalizing this for the $k$-th couple and summing the probabilities we get
\[\sum_{k=1}^{m}\frac{12\lambda(k-1) \log n}{(\lambda - 1)n} < \frac{\lambda m^2 \log n}{(\lambda - 1)n} = O \left(\frac{\log n}{n^{2 \delta}}\right)\]
which goes to zero as $n$ goes to infinity. Note that if $m = \sqrt{n}$ this argument would not hold. In fact an  argument similar to the Birthday Paradox shows the direct algorithm fails with high probability if the number of couples is a large multiple of $\sqrt{n}$.

\paragraph{}
\noindent{\bf 2. Number of couples is  $n^{\frac{2}{3}-\delta}$:}   The direct algorithm algorithm attempts to insert the couples according to a single permutation. A natural attempt to find a stable matching when more couples are in the market is to change the permutation each time a couple kicks out another couple.  Consider the following addition to the direct algorithm: each time a couple $c_i$ evicts a different couple $c_j$ the algorithm starts over but swaps the order between $c_i$ and $c_j$ when the couples apply.

Denote the initial order of insertion by $c_1, c_2, \ldots c_m$. If $c_i$ evicts $c_j$ for $i > j$, swapping places between $c_i$ and $c_j$  will cause $j$ not to be evicted by $c_i$. However, this could create new ``evictions". One can prove that the probability that any other couple ``feels" that $c_i$ and $c_j$ have  swapped places in the application order is at most $O(n^{-1/3-\delta/2})$.
By a similar analysis as in the direct algorithm, the probability that any of the doctors who got evicted  by $c_i$ or $c_j$ enters any of the hospitals of these couples is bounded by
    \[\frac{24n^{2/3-\delta}\log n}{n} < n^{1/3-\delta/2}.\]

What is left to bound is the number of swaps; again, the probability that $c_k$ will evict another couple is roughly $\frac{k}{n}$ where we neglect the $\log n$ factor. Thus the expected number of couples which will evict another couple is bounded by $\frac{m^2}{n} < n^{1/3 - \delta}$. Informally, combining these together one obtains that with probability approaching 1 swapping will solve all  ``eviction" events, the algorithm will find a stable matching and will terminate successfully.

Note that implicitly we assumed  here that any pair of couples will swap places at most once. Unfortunately this approach is not sufficient and to formally obtain our result and we will need some more subtle structures.

\paragraph{}
\noindent{\bf 3. Number of couples is  $n^{1-\epsilon}$  (sketch of proof of Theorem \ref{thm:stable}):}

The SoDA algorithm attempts to find an ordering of the couples, such that if couples apply one by one according  to this order, no couple gets evicted by another couple. Whether or not a couple $c$  evicts another couple $c'$ depends on the (current) matching and  the preference profile. Identifying worst case scenarios, such as where $c$ could ``possibly" evict $c'$ if there exist a configuration in which this happens, are too weak to prove our result. Instead, we devise a notion of whether $c$ is ``likely" to evict $c'$, and use this notion to analyze the algorithm. To do so we define for each couple $c$ an {\it influence tree}; roughly speaking the influence tree of $c$ consists of the  hospitals and doctors which $c$ is most likely to influence, i.e. the doctors and hospitals that are likely to obtain new matches due to the presence of $c$ and its likely evictions.

We will want to show that there are not  ``many" influence tree intersections, since intersections imply two couples might be able to influence the same hospital, and more importantly they might evict each other.  A first key step in this direction is the following:

\noindent{\bf (i) With high probability each influence tree is small (with respect to $n$).}

If influence trees had not  intersected each other, one could have shown that any insertion order of the couples would yield a stable matching with high probability. Essentially \citet{RothCouples} showed that if $\epsilon<0.5$ then the probability that no two influence trees intersect approaches $1$ as $n\to\infty$. This however is not the case for all $\epsilon<1$.

Influence trees, their intersections and hospital preferences induce a useful structure in the form of a directed graph which we call the {\it couples graph}; roughly speaking, in the couples graph each couple is a node, and there is a directed edge from couple $c$ to another couple $c'$ if their influence trees intersect at some hospital $h$ and $c$ can possibly evict some doctor that caused $h$ to be in the influence tree of $c'$ (the doctor can be a member of the couple $c'$).  We will show that the couples graph is sparse:

\noindent{\bf (ii) With high probability all weakly connected components in the couples graph are small.}\footnote{A weakly connected component in directed graph is a connected component in the graph obtained by removing the directions of the edges.}

Recall that an influence tree for one couple does not involve other couples. In the next  step we  verify that influence trees are indeed the ``right" structure:

\noindent{\bf (iii) With high probability if in the algorithm a couple $c$ influences a hospital $h$ under any ordering over the couples $\pi$, then that hospital will also belong to the influence tree of $c$.}

Finally, if one can find a topological sort $\pi$ in the couples graph\footnote{A topological sort $\pi$ is an order over the couples such that no couple has an edge to a couple ahead of him in the order.}  then by letting couples apply one by one according to $\pi$  yields a stable matching. We show:

\noindent{\bf (iv) With high probability there are no directed cycles in the couples graph.}





\subsection{Proof of Theorem \ref{thm:stable}}


We begin with defining influence trees. These will be defined for a fixed realization of the preferences and with respect to a parameter $r$ which should be interpreted as ``possible rejections". First we need a few notations.
Let $\Gamma=(H,S,C,\succ_H,\succ_S,\succ_C)$ be a matching market and let $\mu$ be a matching.
Denote by $o_h(\mu)$ and by $f_h=k_h-o_h(\mu)$ the number of assigned doctors to hospital $h$ and the number of available positions in $h$  under $\mu$ respectively.
We also denote by $d^j(\mu,h)$ to be the $j$-th least preferred doctor according to $\succ_h$ that is assigned to $h$ under $\mu$.


\begin{definition}[Influence Tree]
Let $\Gamma=(H,S,C,\succ_H,\succ_S,\succ_C)$ be matching market with couples and let $\mu$ be the matching produced by the DA algorithm for the sub-market without couples.
Let $d\in D$ and let $r$ be any integer. An {\bf influence sub-tree} of doctor $d$ with root $h$ and with respect to $r$ , denoted by $IT(d,r,h)$ is defined recursively as follows.
\begin{enumerate}
\item[(a)]
If $f_h(\mu)=0$ and $d^{k_h}(\mu,h)\succ_h d$ then let $h'$ be be the next preferred hospital by $d$ after $h$ and let $IT(d,r,h)=IT(d,r,h')$. Otherwise

\item[(b)] Change $\mu$ such that $d$ is assigned to $h$ and:

\begin{enumerate}
\item[(b1)] Add $(h,d)$ to $IT(d,r,h)$.

\item[(b2)] If $r>0$ or $f_h(\mu)=-1$ then:  for each $j=1,\ldots, \min(o_h(\mu),r-f_h(\mu))$ let $h_j$ be the most preferred hospital by $d^j(\mu,h)$ after $h$, and add to $IT(d,r,h)$ the influence sub-tree $IT(d^j(\mu,h),r-(j-1)-f_h(\mu),h_j)$.

\end{enumerate}
\end{enumerate}
For a couple $c=\{f,m\}$, let $(h_f^1,h_m^1),\ldots,(h_f^r,h_m^r)$ be the top $r$ pairs of hospitals according to $\succ_c$ in which the couple $c$ can be accepted. That is, either
\begin{itemize}
\item
$h_f^i=h_m^i$ and $c\subseteq Ch_{h_f^i}(\mu(h_f^i)\cup c)$, or
\item
$h_f^i\neq h_m^i$ and $f\in Ch_{h_f^i}(\mu(h_f^i)\cup \{f\})$ and $m\in Ch_{h_m^i}(\mu(h_m^i)\cup \{m\})$.
\end{itemize}
The influence tree for the couple $c$ is defined to be:
$$
IT(c,r):=\bigcup_{i=1}^r \left(IT(f,r+1-i,h_f^i))\cup IT(m,r+1-i,h_m^i) \right).
$$
\label{def:influence-tree}
\end{definition}

First note that we allow $f_h(\mu)$ to be -1 in the definition of an influence tree (this is possible since under this definition we first assign a doctor to a hospital and only then reject from that hospital.)
Also observe that each time a hospital $h$ is inserted to the influence tree, a doctor $d$ is associated with it.
In this case we say that  $h$  was {\it inserted} to $IT(c,r)$ {\it by} $d$.\footnote{We do not rule out here that $h$ was inserted to the influence tree by two different doctors. We will later
show that the probability of this even is negligible, however.} With a slight abuse of notation we will write $h\in IT(c,r)$ if there exist a doctor $d$ such that $(h,d)\in IT(c,r)$, i.e. $h\in IT(c,r)$ is inserted to $d$ by some doctor.

In the definition of an influence tree for $c$, no other couple other than $c$ involved; the definition in fact simulates the presence of other couples, or in other words it simulates an adversary that can ``reject" doctors  from settling in a hospital $h$ due the possible additional occupied positions that will possibly be taken due to the presence of other couples. The adversary is allowed to reject $r$ times (above the natural rejections).   Importantly, Definition \ref{def:influence-tree}   allows us to analyze a static setting rather than a dynamic setting in which at each point a different number of couples already applied.

Before we continue with the proof we illustrate the definition of an influence tree in the following example.
\begin{example}
Consider a setting with 6 hospitals each with capacity of 2, 5 single doctors, $d_1,d_2,\ldots,d_5$ and two couples $c_1=(d_6,d_7)$ and $c_2=(d_8,d_9)$, and let their preferences be as in Table \ref{table:exInfluence}. To simplify the illustration we chose a preferences that does not ``seem" to be drawn randomly.
\begin{table}[htb]
\begin{center}
\begin{tabular}{ccccccccccccc}
\hline\hline
 \multicolumn{7}{c}{doctors} & & \multicolumn{5}{c}{hospitals}\\
\cline{1-7} \cline{9-13}
 $d_1$ & $d_2$  & $d_3$ & $d_4$ & $d_5$ & $(d_6,d_7)$ & $(d_8,d_9)$  & &  $h_1$ & $h_2$ &  $h_3$ & $h_4$ &  $h_5$  \\
\cline{1-7} \cline{9-13}\\*[-12pt]
 \fbox{$h_1$} & \fbox{$h_1$} & $h_1$  &  \fbox{$h_3$} & \fbox{$h_3$} & $(h_1,h_2)$  & $(h_1,h_1)$   & &  \fbox{$d_1$} & $d_1$ & $d_1$ & $s_1$  & $s_1$  \\
 $h_2$ & $h_2$ & \fbox{$h_2$}  &  $h_5$ & $h_5$& $(h_2,h_1)$  & $(h_2,h_2)$   & & $d_8$ & $d_8$ & $d_8$ & $d_8$  & $d_8$ \\*[2pt]
 $h_3$ & $h_3$ & $h_3$  &  $h_1$ & $h_1$ & $(h_3,h_4)$  & $(h_3,h_4)$   & & $d_9$ & $d_9$ & $d_9$ & $d_9$  & $d_9$ \\*[3pt]
 $h_4$ & $h_5$ & $h_4$  &  $h_4$ & $h_2$ & $(h_4,h_5)$  & $(h_4,h_3)$   & & \fbox{$d_2$} & $d_2$ & $d_3$ & $d_3$  & $d_6$ \\*[3pt]
 $h_5$ & $h_6$ & $h_5$  &  $h_2$ & $h_4$ & $(h_5,h_5)$  & $(h_4,h_2)$   & & $d_5$ & $d_5$ & $d_6$ & $d_5$  & $d_4$ \\*[3pt]
 &     &       &  &   &    & &   &                                   $d_3$ & \fbox{$d_3$}  & $d_2$ & $d_4$  & $d_2$ \\*[3pt]
 &     &       &  &   &    & &    &                                  $d_6$ & $d_6$ & \fbox{$d_5$} & $d_6$ & $d_5$   \\*[3pt]
  &     &       &  &   &    & &   &                                   $d_4$ & $d_4$  & $d_7$ & $d_2$  & $d_7$ \\*[3pt]
 &     &       &  &   &    & &    &                                  $d_7$ & $d_7$ & \fbox{$d_4$} & $d_7$ & $d_3$   \\*[3pt]
 \hline\hline
\end{tabular}
\end{center}
\caption{Preference lists.}
\label{table:exInfluence}
\end{table}

The Deferred Acceptance algorithm for the market without couples produces the matching given in the boxes as in Table \ref{table:exInfluence}.
The influence trees of $c_1=(d_8,d_9)$ with parameters $r=0$ and  $r=1$ are given in Figure \ref{fig:inf}(a). For $r=0$ the tree captures the ``chain reaction" that $c_1$ causes after entering the first pair of hospitals  that accepts it, these the pair of hospitals $(h_3,h_4)$. For $r=1$, the tree Had $c_1$ would be rejected from the pair $(h_3,h_4)$ note that the next pair that would have accepted it would be $(h_4,h_5)$.
Thus the influence tree of $c_1$ includes with $r=1$ includes both its tree for $r=0$ and the chain reaction it causes had it been accepted to $(h_3,h_4)$ (see Figure \ref{fig:inf}(b)).
Similarly the influence tree of couple $c_2=(d_8,d_9)$ is given in  Figure \ref{fig:inf}.
\begin{figure}[htb]
\centering
\subfloat[Influence tree of $c_1=(d_6,d_7)$.]{\label{fig:withHold}\includegraphics[width=6cm,height=2cm]{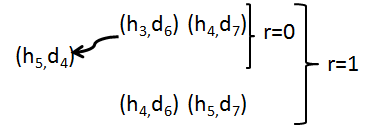}}
\subfloat[Influence tree of $c_2=(d_7,d_8)$.]{\label{fig:AllwithHold}\includegraphics[width=7cm,height=2cm]{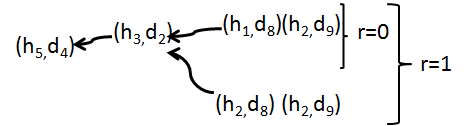}}
\caption{Influence trees with parameters $r=0$ and $r=1$.}
\label{fig:inf}
\end{figure}
\label{example-inf}
\end{example}

At this point we fix $r$ to be  $r = 4/\epsilon$ for some fixed $0<\epsilon<1$. One should interpret this $r$ as a ``small" number of possible rejections (relative to $n$). In random market the influence trees (IT's) are random variables.
\begin{lemma}
\begin{enumerate}
\item For every hospital $h$ couple $c$, $\Pr(h\in IT(c,r))=O\left((\log n)^{r+1} / n\right)$.
\item The probability that the size of every influence tree $IT(c,r)$ is $O((\log n)^{r+1})$ is at least  $1-n^{-3}$.
\item
The probability that for all couples $c$, each hospital $h$ appears in $IT(c,r)$ at most once is at least $1-n^{-\epsilon/2}$.
\end{enumerate}
\label{small-blob}
\end{lemma}
\begin{proof}
We begin with the second part. Let $c$ be a couple. For each of the two $d\in c$ and for each $h'\neq h$ we will
give an upper bound of $O\left((\log n)^r / n\right)$ on $\Pr(h\in IT(d,r,h'))$. The claim will follow from the definition of $IT(c,r)$ and
union bound.

An alternative way of viewing the recursive definition of $IT(d,r,h')$, is as follows: doctor $d$ proceeds down his list beginning with $h'$ until
he finds the first hospital willing to accept him. If $d$ is accepted into a hospital $h_1$ and $h_1$ was full to capacity, then some doctor $d'$ is evicted
and goes to a hospital $h_2$, and we add $IT(d',r,h_2)$ to $IT(d,r,h')$. In this case, continuing the ``chain reaction" did not require any arbitrary rejections.
We call the hospitals added into $IT(d,r,h')$ with parameter $r$ {\bf the main path} of $IT(d,r,h')$.
We then also allow the adversary to introduce up to $r$ arbitrary rejections (for example, precluding $d$ from being accepted into $h_1$).
Thus the influence tree is composed of the main path, with lower-order  influence trees (i.e. influence trees with a strictly smaller value of $r$) attached along it.

 We first show by induction that with probability at least $1-n^{-6}$ the length of the main path in $IT(d,r,h')$ is at most $b\log n$, where $b=6\cdot\frac{c_{max}\cdot \gamma_{max}}{\lambda-1}$.
 At any step along the main path, for the main path to continue, the currently evicted doctor $d$ needs to choose a {\em full} hospital $h$.
 Because of the way the doctors' preferences are sampled, the probability of this happening is bounded by $1-\frac{\lambda-1}{c_{max}\cdot \gamma_{max}}$. Since each subsequent step along the path is independent
 from the previous ones, the bound follows.

 By union bound, we see that with probability at least $1-n^{-4}$ {\em all} potential main paths contain at most $b\log n$ hospitals. Each main path of length $\ell$ recursively
 gives rise to at most $r\cdot \ell$ lower-order influence trees (i.e. influence trees with smaller $r$) that are added to $IT(d,r,h')$. Thus we can prove by induction that for each
 $r$, the size $S(r)$ of the largest order-$r$ influence tree is bounded by $(1+b r \log n)^{r+1} = O((\log n)^{r+1})$. For the base case, an influence tree with $r=0$ only contains the main path,
 and thus $S(0)\le b\log n$. For the step, we get
 \begin{multline*}
 S(r)\le b\log n + (b\log n)\cdot r \cdot S(r-1) \le b\log n +  (b\log n)\cdot r \cdot (1+b r \log n)^{r} < \\
  (1+b r \log n)^{r} +  (b\log n)\cdot r \cdot (1+b r \log n)^{r} = (1+b r \log n)^{r+1}.
 \end{multline*}

Next, the first part of the lemma follows from the proof of the second part and the fact that the hospitals that are added to $IT(c,r)$ are hospitals on
the  doctors' preference lists and are chosen independently. Thus the probability of $h$ to be added to $IT(c,r)$ at some point is bounded
by $S(r)\cdot (c_{\max}\cdot \gamma_{max}/n) = O((\log n)^{r+1}/n)$.

Finally, we show that $IT(c,r)$ does not ``intersect itself" except with probability $<n^{\epsilon/2}$. Note that in particular this means that
the members of the couple may not apply into the same hospital or evict each other. We have seen that the probability of a hospital $h$ belonging
to $IT(c,r)$ is bounded by $O(S(r)/n)$. Similarly, the probability of $h$ to be added twice or more to $IT(c,r)$ is bounded by $O(S(r)^2/n^2)$.
Taking a union bound over all possible hospitals $h$ and all possible couples $c$, we see that the probability that any hospital appears in
any $IT(c,r)$ twice or more is bounded by
$$
O(S(r)^2/n^2) \cdot n \cdot n^{1-\epsilon} < n^{-\epsilon/2}.
$$
\end{proof}

Throughout the remainder of the proof we will assume that each hospital appears in each $IT(c,r)$ at most once, neglecting an event of probability $<n^{-\epsilon/2}$.

In fact, in Lemma \ref{small-blob}, one can prove a stronger bound of $O(\log n /n)$ for the probability that a hospital belongs to an influence tree. Although we do not prove or use the stronger bound in the rest of the paper, it provides intuition for why the SoDA algorithm works well in even in a rather small market (e.g. when $n = 256$ we have $(\log 256)^3 =8^3 = 512$ which does not explain why the algorithm works).


Next we analyze how much influence trees intersect with each other.
Let $c_1$ and $c_2$ be two different couples. We say that two influence trees $IT(c_1,r)$ and $IT(c_2,r)$  {\bf intersect} at hospital $h$ if there exist $d'$ and $d''$ such that $d'\neq d''$, $(h,d')\in IT(c_1,r)$ and $(h,d'')\in IT(c_2,r)$.\footnote{It  is possible that  if two influence trees intersect they will have other nodes  $(\tilde{h},\tilde{d})$ in common, since there might be common paths that continue from the point they intersect.}

\begin{lemma}
No two influence trees intersect more than once, except with probability $<n^{-\epsilon/2}$.
\end{lemma}
\begin{proof}
By Lemma \ref{small-blob}, we can assume that for  every couple $c$ the size of $IT(c,r)$ is at most $O\left((\log n)^{r+1}\right)$.
For the remainder of the proof, we will denote this upper bound on the size of $IT(c,r)$ by $S(r)=O\left((\log n)^{r+1}\right)$. Recall also that we have assumed that no $IT(c,r)$ intersects itself.

We prove that with high probability no two influence trees intersect exactly $2$ times. A similar proof shows that for every $3\le k \le S(r)$ no two influence trees
intersect exactly $k$ times. The proof will then follow by a union bound on $k$ (since the size of each tree is  $\le S(r)$ with high probability they cannot intersect more than $S(r)$ times).

  Let $c_1, c_2$ be two couples, and $h_1, h_2$ be two hospitals. We want to bound the probability of the event
  \begin{equation}
  \label{eq:1}
  \Pr(h_1, h_2 \in IT(c_1,r) \cap IT(c_2,r)) = \Pr(h_1, h_2 \in IT(c_1,r)) \cdot \Pr(h_1, h_2 \in IT(c_2,r)| h_1, h_2 \in IT(c_1,r)).
  \end{equation}
We first note that if $h_1$ is an ancestor of $h_2$ in, e.g. $IT(c_1,r)$, and $IT(c_1,r)$ intersects $IT(c_2,r)$ in both
$h_1$ and $h_2$, then the influence tree $IT(c_2,2r+c_{max})$ will self-intersect at $h_2$. The hospital $h_2$ will be added to $IT(c_2,2r+c_{max})$ twice:
once following the path in $IT(c_2,r)$, and a second time through $h_1$ and then following the path from $h_1$ to $h_2$ in $IT(c_1,r)$. Since
$2r+c_{max}$ is a constant, by Lemma \ref{small-blob} the probability that any $IT(c,2r+c_{max})$ will self intersect is smaller than $n^{-\epsilon/2}$, and can be disregarded.
Thus we can assume that $h_1$ and $h_2$ are not each other's ancestors in either $IT(c_1,r)$ or $IT(c_2,r)$.

We begin by calculating the probability of the first event in \eqref{eq:1}. A similar proof to that of Lemma \ref{small-blob} gives that the probability for this event is \[\Pr(h_1, h_2 \in IT(c_1,r)) = O\left(\frac{S(r)^2}{n^2}\right).\]

  Rather than compute $\Pr(h_1, h_2 \in IT(c_2,r)| h_1, h_2 \in IT(c_1,r))$ directly, to avoid the conditioning, we consider inserting $c_2$ into a modified world, in which all  hospitals in $IT(c_1,r)$ except for $\{h_1,h_2\}$ and all the doctors in these hospitals do not exist. We argue that in this case,
    \[\Pr(h_1, h_2 \in IT(c_2,r)) = O\left(\frac{S(r)^2}{n^2}\right)\]
  using similar reasoning.

The influence tree generated in the modified algorithm (where we took out some of the hospitals) may differ from the one in the ``real" algorithm. Note however that if removing $IT(c_1,r)$ affects the generation of the tree $IT(c_2,r)$ before it reaches $h_1, h_2$, then it is the case that $IT(c_2,r)$ intersects $IT(c_1,r)$ at another hospital (which comes before $h_1,h_2$). But this is a contradiction, since we assumed $IT(c_1,r), \ IT(c_2,r)$ intersect {\em exactly} twice.

  Multiplying the probabilities, we get that
  \[\Pr(h_1, h_2 \in IT(c_1,r) \cap IT(c_2,r)) = O\left(\frac{S(r)^4}{n^4}\right)\]
  Taking a union bound over $O(n)$ hospitals and $n^{1 - \epsilon}$ couples, bounds the probability that exist two couples which intersect exactly twice is at most
  \[ O\left(\frac{S(r)^4}{n^{2\epsilon}}\right).\]
  We do not present the proof for exactly $k$ intersections, and only state that the probability for that event drops at a rate of
  \[\frac{S(r)^{2 k}}{n^{k \cdot \epsilon}} < \frac{S(r)^4}{n^{2\epsilon}}.\]
  Taking a union bound over all possible values of $k$, we get that the probability that any two couples intersect strictly more than once is at most
  \[ O\left(\frac{S(r)\cdot S(r)^4}{n^{2\epsilon}}\right) = \frac{\polylog(n)}{n^{2\epsilon}}.\]
  as required.\footnote{We write $\polylog n$ for a polynomial in $\log n$. In particular $\frac{\polylog n}{n^{2\epsilon}}$ tends to zero as $n$ tends to infinity.}
\end{proof}

Observe that in the definition of an influence tree for a couple $c$, no other couple is involved and therefore the tree captures only what possibly could have happened  had there  been other couples.  The SoDA algorithm inserts couples one by one after the DA algorithm has terminated, and if some couple $c_1$ evicts  another couple $c_2$  the order of their insertions is altered so that $c_1$  is moved ahead of $c_2$.
Intuitively the intersection of two  influence trees, of $c_1$ and of $c_2$, together with the hospital preferences will provide a good guess which couple to insert first. This motivates the following definition of the couples graph:

\begin{definition}
Let $\Gamma=(H,S,C,\succ_H,\succ_S,\succ_H)$  be a matching market and let $r>0$.
In a (directed) {\bf couples graph} for depth $r>0$, denoted by $G(C,r)$ the set of vertices is $C$ and for every two couples  $c_1,c_2\in C$
there is a directed edge from $c_1$ to $c_2$ if and only if there exist $h\in H$ and $d_1,d_2\in D$ ($d_1\neq d_2$) such that $(h,d_1)\in IT(c_1,r)$ and $(h,d_2)\in IT(c_2,r)$ and $d_1\succ_h d_2$.
\label{def:couples-graph}
\end{definition}
Before we continue we illustrate a couples graph.
\begin{example}
Consider the same market as in Example \ref{example-inf} (see Table \ref{table:exInfluence}).
Note that the influence trees with $r=1$  intersect in $h_3$ where $(h_3,d_2)\in IT(c_2,1)$ and $(h_3,d_6)\in IT(c_2,1)$.
Since $s_6\succ_{h_3} d_2$ the  couples graph with $r=1$ is as in Figure \ref{fig:graph}. Indeed letting $c_1$ apply before $c_2$ (after the DA stage) will end without any couple evicting each other and in a stable matching.
\begin{figure}[htb]
\centering
{\includegraphics[width=3cm,height=1cm]{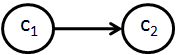}}
\caption{Couples graph for $r=1$.}
\label{fig:graph}
\end{figure}
\label{example-graph}
\end{example}

Our goal will be to show that with high probability the graph $G(c,r)$ can be topologically sorted; such a sorting corresponds to a ``good" insertion order of the couples in the SoDA algorithm. In example \ref{example-graph} the order $c_1,c_2$ is a topological sort.

In a couples graph $G=G(C,r)$ a {\it weakly connected  component} is defined to be a connected component in the graph obtained from $G$ by removing the direction of the edges.\footnote{A set of nodes in an undirected graph is called a connected component if there exists a path between each to nodes in the set.}
\begin{lemma}\label{small-weak-connected-components}
With  probability  $>1-1/n$ the largest weakly connected component of the couples graph has size at most $\frac{3}{\epsilon}$.
\end{lemma}

\begin{proof}
   We will first consider an arbitrary set of  $\frac{3}{\epsilon}$ couples and show that the probability that they form a weakly connected component is very small.
   The statement of the lemma will follow through union bound.
    Let $I = \left(c_{1},c_{2}, \ldots,  c_{\lfloor 3/\epsilon\rfloor}\right)$ be a sequence  of couples with no repetitions: $c_i \neq c_{j}$. Let $A_{I}$ be the event that for every $1<i\leq \lfloor 3/\epsilon\rfloor$ the influence tree of $c_i$ intersects with one of the previous influence trees, that is
  \[IT(c_{i},r) \bigcap \left(\cup_{j < i} IT(c_{i},r) \right) \neq \emptyset.\]

 We first show that
  \begin{equation}
  \Pr(A_{I}) \le \frac{(S(r)^2\cdot c_{max} \cdot \ga_{max} \cdot 3/\epsilon)^{\lfloor 3/\epsilon \rfloor}}{n^{\lfloor 3/\epsilon\rfloor-1}}\le \frac{(S(r)^2\cdot c_{max} \cdot \ga_{max} \cdot 3/\epsilon)^{3/\epsilon}}{n^{ 3/\epsilon-2}},
  \label{ineq:AI-small}
  \end{equation}
  where $S(r)$ is the bound on the size of the influence trees $IT(c_i,r)$ as in  Lemma \ref{small-blob}.

%
%

Let
    \[IT_i =  \cup_{j \le i} IT(c_j, r)\]
    be the union of the influence trees of the first $i$ couples. The probability of $A_{I}$ can be written as

\begin{equation*}
\Pr(A_{I}) = \Pr \left(IT(2,r)\cap IT_{1} \neq \emptyset\right) \cdot \Pr\left(IT(3,r)\cap IT_{2} \neq \emptyset| \ IT(2,r)\cap IT_{1} \neq \emptyset\right)\cdot \cdots
\end{equation*}
\begin{equation}
\Pr\left(IT(\lfloor 3/\epsilon \rfloor ,r)\cap IT_{\lfloor 3/\epsilon \rfloor - 1} \neq \emptyset| \forall j \le \lfloor 3/\epsilon-1 \rfloor, \ IT(j,r)\cap IT_{j-1} \neq \emptyset \right).
\label{eq-AI-ind}
\end{equation}

All the interactions that cause the influence trees within $IT_{j-1}$ to intersect happen within $IT_{j-1}$, and conditioned on the set $IT_{j-1}$ of hospitals do not affect the probability
of $IT(c_j,r)$ intersecting $IT_{j-1}$.
Hence for every $j=2,\ldots,\lfloor 3/\epsilon\rfloor$,
\begin{multline*}
\Pr\left(IT(c_j ,r)\cap IT_{j - 1} \neq \emptyset| \forall 2\leq l \le j-1, \ IT(l,r)\cap IT_{l-1} \neq \emptyset \right) = \\ \Pr\left(IT(c_j ,r)\cap IT_{j - 1} \neq \emptyset~|~IT_{j-1}\right).
\end{multline*}

Furthermore from  Lemma \ref{small-blob} it follows that the probability that $|IT(c_l,r)|<S(r)$ is at least $1-\frac{1}{n^3}$ and therefore $|IT_{j}|<j\cdot S(r)$. Hence,
\begin{equation*}
\Pr\left(IT(c_j ,r)\cap IT_{j-1} \neq \emptyset ~|~IT_{j-1}\right)\leq \frac{(j-1)\cdot S(r)^2\cdot \ga_{max}}{\lambda n/c_{max}}+\frac{1}{n^3} <\frac{j\cdot S(r)^2\cdot \ga_{max}}{\lambda n/c_{max}}.
\end{equation*}
Since there are $\lfloor 3/\epsilon\rfloor-1$ terms in (\ref{eq-AI-ind}) we derive inequality (\ref{ineq:AI-small}).

To finish the proof, observe that if there is a connected component of size at least $3/\epsilon$ then there exists a sequence $I$ such that $A_{I}$ holds. Since there are $n^{1-\epsilon}$ couples there  exists
fewer than \[\left(n^{1 - \epsilon}\right)^{3/\epsilon} = n^{3/\epsilon - 3}\] such possible sequences $I$. Therefore using a union bound over all of them proves the lemma.
\end{proof}

Recall that we ignore all realizations of preferences at which two influence trees intersect more than once (in particular there is at most a single edge between every two couples in the couples graph). From now one we also ignore realizations where the largest weakly connected component of the couples graph contains more than $3/\epsilon$ couples.
\begin{lemma} \label{no-cycle}
With probability $1- O\left(\frac{1}{n^{\epsilon}}\right)$ the couples graph has no directed cycles.
\end{lemma}
\begin{proof}
We first prove the following claim, that is basically a simple general statement about directed graphs:
\begin{claim} \label{k-hospitals-in-minimal} If the shortest directed cycle has length $k$, it involves $k$ different hospitals.
\end{claim}
\begin{proof}
  Suppose the shortest directed cycle is of length $k$ and consider such a cycle $c_1 \rightarrow c_2 \rightarrow \cdots \rightarrow c_k \rightarrow c_1$. Suppose couples $c_1$ and $c_2$ intersect at $h$ due to $d_1$ and $d_2$ respectively, i.e. $(h,d_1) \in IT(c_1,r)$, $(h,d_2) \in IT(c_2,r)$  and $(h,d_2) \in IT(c_2,r)$. Assume for contradiction that for some $2\leq i\leq k$, $c_{i}$ and $c_{i+1}$ (i is taken modulo k) intersect at hospital $h$ due to some doctors $d_i$ and $d_{i+1}$, i.e. $(h,d_{i}) \in IT(c_i,r)$, $(h,d_{i+1}) \in IT(c_{i+1},r)$ and $d_i \succ_h d_{i+1}$. Consider the case in which $d_{i} \succ_{h} d_{2}$. In this case a cycle of length less than $k$ exists which consists of $c_2 \rightarrow c_3 \rightarrow \cdots \rightarrow c_i \rightarrow c_2$. If  $d_{2} \succeq_{h} d_{i}$, i.e.  either $d_2 \succ_{h} d_{i}$ or $d_2 = d_i$, then $d_1 \succ_{h} d_2 \succeq_{h} d_{i} \succ_{h} d_{i+1}$ implying that $c_{1} \rightarrow c_{i+1} \rightarrow \cdots \rightarrow c_k \rightarrow c_{1}$ is a shorter cycle.
\end{proof}

To prove the lemma it is sufficient to show that the probability that the shortest directed cycle has length $k$ is $O\left(\frac{S(r)^{2 k}}{n^{\epsilon k}}\right)$ since by taking the sum of these probabilities over all values of $k$ gives the result (note that the the dominant term in this sum is when $k=2$).

We proceed in a manner similar to that of the proof of Lemma \ref{small-weak-connected-components}. Let $I = \left(c_{1},c_{2}, \ldots,  c_{k}\right)$ be a sequence of couples without repetitions $c_i \neq c_{j}$. Let $J = \left(h_{1},h_{2}, \ldots,  h_{k}\right)$ be a sequence of $k$ hospitals without repetitions $h_i \neq h_j$. Let $A_{I,J}$ be the event that for every $i=1,\ldots,k$, $IT(c_i,r)$ and $IT(c_{i+1},r)$ intersect at hospital $h_i$. Applying Lemma \ref{small-blob}, and using reasoning  similar to the proof of Lemma \ref{small-weak-connected-components} the probability of the event $A_{I,J}$ can be bounded by
\[\Pr(A_{I,J}) < \frac{(2 S(r)\cdot \ga_{max})^{2k}}{(\lambda n/c_{max})^{2k}}.\]
 Since there are $\le \lambda n$ positions and $n^{1-\epsilon}$ couples, there are $\lambda^k n^k n^{(1-\epsilon)k}$ such different events $A_{I,J}$. A  union bound over all these events implies the lemma.
\end{proof}

For the analysis we will consider the event that the couples graph contains  a cycle as a failure.\footnote{The presence of a cycle does not necessarily imply  that there is no stable matching. In fact the SoDA will often find  stable matchings even when there are cycles in the couples graph.} If the couples graph does not  have cycles, then it has a topological sort. Let $\pi$ denote any topological sort of $G$. We claim that
inserting the couples according to $\pi$ will result in a stable matching with couples. Moreover, we will show that a failure of the SoDA algorithm corresponds to a backward edge in the couples graph.\footnote{A
backward edge is an edge from a newly inserted couple to a previously inserted one.}


The next lemma shows that the influence trees indeed captures ``real influences".
\begin{lemma}\label{same-insertions}
Suppose we insert the couples as in the SoDA algorithm according to some order  $\pi$ until a couple evicts another couple or until all couples have been inserted. If a couple $c$ is inserted and influences  hospital $h$ , then  $h \in IT(c,r)$.
\end{lemma}
\begin{proof}
Recall that we consider only ``small" weakly connected components (Lemma \ref{small-weak-connected-components} upper bounded the probability that such a component is large).
Let $c$ be the couple currently being inserted, and assume that the statement of the lemma was true for couples inserted before $c$.
Let  $\{c_1, \ldots c_k\}$ be $c$'s weakly connected component in the couples graph, where $k \le 3/\epsilon$, ordered according to their insertion order in $\pi$.
We prove by induction a stronger claim, namely that if  $c=c_i$ influenced a hospital $h$, then $h \in IT(c, i-1)$.

Suppose that $c=c_i$ is currently being inserted and that its insertion affects a hospital $h$. Consider the path of evictions that was started
by $c$ and led to hospital $h$ being affected. There are two types of evictions along this path: the first type would have occurred even without
any other couples present. The second type occurs because a hospital $h'$ on the path has already been affected by a previously inserted couple $c_j$. If this
happens, then the influence tree of $c$ intersects the influence tree of $c_j$ and thus $c_j$ in in the weakly connected component of $c$ in the couples
graph. Moreover, since influence trees intersect only once, evictions due to influences from previously inserted couples happen at most $i-1$ times:
at most once for each previously inserted couple in the weakly connected component of $c$. By the definition of $IT(c,i-1)$ this implies $h\in IT(c,i-1)$.



\end{proof}

As an immediate corollary of Lemma \ref{same-insertions} we obtain that a couple causing another couple to be
evicted corresponds to an edge in the couples graph.

\begin{corollary}
\label{col:ins1}
If in an insertion order $\pi$ inserting the couple $c_{\pi(i)}$ causes the couple $c_{\pi(j)}$ to be evicted ($j<i$) then
in the couples graph there is an edge from $c_{\pi(i)}$ to $c_{\pi(j)}$.
\end{corollary}

Since there exist a topological sort with a high probability Theorem \ref{thm:stable} follows from the following corollary:
\begin{corollary}
\label{cor:ins2}
Inserting the couples according to any topologically sort  $\pi$ of the couples graph gives a stable outcome.
\end{corollary}

Finally, we can now analyze the running time of (a slight modification of) the SoDA algorithm. Note that with high probability
we have that the couples graph has small connected components (of size $<3/\epsilon$) and can be topologically sorted.
According to Corollary \ref{cor:ins2} each failed iteration of the SoDA algorithm is due to a backward edge in the
insertion order $\pi$. By recording the backward edge, and ensuring that all future attempts are consistent with it, we
can guarantee that at most $(3/\epsilon)^2\cdot n^{1-\epsilon}$ permutations will be tried before either a topologically
sorted order is arrived at, or a cycle in the couples graph is found.\footnote{It can be shown that the SoDA algorithm without
this modification will run with at most $(3/\epsilon)^{3/\epsilon}\cdot n^{1-\epsilon}$ iterations.}


\section{Incentive Compatibility}

In this section we will show that:
\begin{theorem} \label{thm:ic}
Ex post truthfulness: The probability that any doctor  can gain by misreporting her preferences is at most $O(n^{-\epsilon/2})$, even if the doctor knows the entire preference list.
\label{truthful}
\end{theorem}
A similar result can be shown for hospitals, using similar techniques as in the proof of Theorem \ref{truthful}. We avoid the exact details here.\footnote{In particular one will need to define influence trees for hospitals, show that with high probability a hospital does not encounter any couple, and with a bit of effort apply Lemma  10 in \cite{KojimaParag} which asserts the desired result for hospitals in markets without couples.}
Together with Theorem \ref{truthful} we obtain that reporting truthfully is a $\delta$-Bayes Nash equilibrium in the Bayesian game induced by the SoDA algorithm (assuming bounded utilities). We refer the reader for exact definitions of the Bayesian game to \cite{RothCouples}.

Throughout this section we will use the same assumptions as in the previous section about the influence trees. They hold except with probability $O(n^{-\epsilon/2})$.
 Informally, we will show that if a doctor or a couple doesn't interact with any other couple's
influence tree, then she does not have an incentive to deviate. To this end we show:

\begin{lemma} \label{lem:ic}
 Let $d\in S$ be any doctor. Suppose that the SoDA algorithm terminates and assigns $d$ to a hospital $h$ in the first (Deferred Acceptance) stage of the algorithm. Suppose that $h$ does not belong to
 any of the couples' influence trees. Then $d$ may not improve her allocation under SoDA by misrepresenting her preferences.

 Similarly, if $c\in C$ is a couple whose influence tree is disjoint from all other influence trees, then
 $c$ may not improve their allocation under SoDA by misrepresenting their preferences.
\end{lemma}

\begin{proof}
We start with the first statement. At the end of the execution of the first stage of the SoDA algorithm $d$ ends up in $h$. By Lemma \ref{same-insertions}, if $d$ was moved from $h$,
in the second stage, then $h$ must belong to the influence tree of one of the couples, contradicting the assumption. Hence at the end of the SoDA algorithm $d$ is still assigned
the hospital $h$.

Suppose that $d$ misrepresents her preferences and obtains a hospital $h'$ such that $h'\succ_d h$ in a valid execution of the SoDA algorithm. It is well known that
the outcome of the (regular) Deferred Acceptance algorithm on singles does not depend on the insertion order. Hence we can execute the SoDA algorithm so that $d$
is the last single doctor to be inserted. Just before $d$ is inserted, for all doctors $d'$ that are assigned to $h'$, $d'\succ_{h'} d$, otherwise $d$ would have been assigned
$h'$ when stating her true preferences. From that point on, a valid execution of the SoDA algorithm does not lead to any couples being evicted, and hence the quality of
the least preferred doctor in $h'$ according to $\succ_{h'}$ may only improve. Hence $d$ may not be assigned to $h'$ in the second phase of the SoDA algorithm. Contradiction.

Next, let $c=(f,m)$ be a couple such that $IT(c,r)$ is disjoint from all other influence trees. Suppose that $c$ is assigned the hospitals $(h_1,h_2)$ is a valid execution of the
SoDA algorithm with an ordering $\pi$ on couples. Since $IT(c,r)$ is disjoint from other influence trees, by Lemma \ref{same-insertions} we see that inserting the couples
in the order $\pi'$ obtained from $\pi$ by putting $c$ first, leads to another valid execution that results in the same allocation.

Suppose that $c$ misrepresent their preferences and obtain the hospitals $(h_1',h_2')\succ_c(h_1,h_2)$ in a valid execution of the SoDA algorithm. Note that
the couple $c$ was the first to be inserted under $\pi'$ and did not get accepted into $(h_1',h_2')$ because one of the hospital preferred all the doctors that
were assigned to it in the DA stage of the algorithm to the corresponding couple member. Without loss of generality, assume that $h_1'$ preferred all of its assigned
doctors to $f$. As in the single doctor case above, in the second phase of the SoDA algorithm the least preferred doctor according to $\succ_{h_1'}$ that is
assigned to $h_1'$ may only improve. Thus $f$ may never be assigned to $h_1'$. Contradiction.
\end{proof}

Using Lemma \ref{lem:ic} we can now prove Theorem \ref{thm:ic}.

\begin{proof} {\bf (of Theorem \ref{thm:ic}).}
Fix any doctor $d\in S$ and the hospital $h$ it is assigned in the DA stage of the SoDA algorithm. By an argument very similar to Lemma \ref{small-blob} we can show that
 the probability that any influence tree contains $h$ (or any other hospital in the influence tree of $d$) is bounded by $O(S(r)^2/n^\epsilon) < n^{-\epsilon/2}$.
 By Lemma \ref{lem:ic}, if this is the case, $d$ does not have an incentive to deviate.

 Similarly, the probability of the influence trees of two couples intersecting is bounded by $O(S(r)^2/n)$, and thus for each couple $c$, the probability that
 $IT(c,r)$ is disjoint from all other influence trees -- and thus $c$ has no incentive to deviate -- is at least $1-O(S(r)^2/n^\epsilon)>1-O(n^{-\epsilon/2})$.
\end{proof}


\section{Simulations}
\label{sec:simul}

In this section we provide simulations results using the SoDA algorithm . In particular we performed sensitivity analysis  on various parameters of the problem. For each configuration we  ran 600 trials. We assumed there are $\frac{n}{2}$ hospitals where $n$ is the number of singles and each hospital has capacity of 3.\footnote{The results can  be slightly improved by  randomizing a new insert  order  each time the algorithm fails (doing this a small arbitrary number of times).}

In the first simulation we fixed the percentage of couples in the market and found the success rate of finding a stable matching.
For comparison, in the NMRP match in 2010  the number of (U.S) doctors was about 16,000 where as the number of couples was about 800.\footnote{In fact in the NRMP more than 20,000 doctors participate, but 16,000 are from the US and are ranked higher in the match.}  As Figure \ref{fig:prectanges} shows that the ratio   of doctors that are members of couples  plays a crucial role in the probability that a stable match will be found.  Note that although the number of singles grows (and the number of couples is linear) the probability for finding a stable match appears to remain unchanged.

\begin{figure}[htb]
\centering
{\includegraphics[width=9.5cm,height=5.5cm]{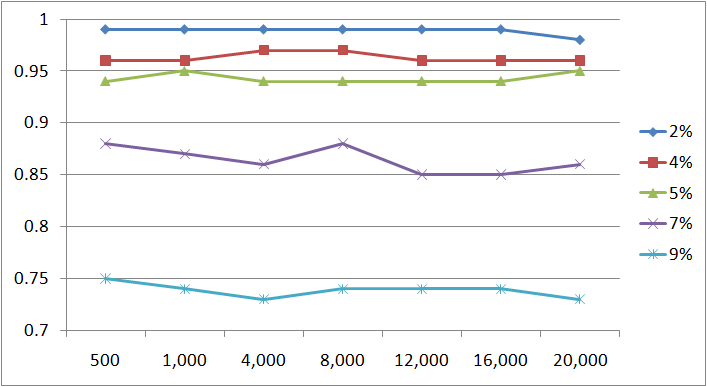}}
\caption{The success rate for finding a stable outcome given the number of singles (x-axis), for different couples percentages (5\% means that 10\% of the doctors are members of couples).}
\label{fig:prectanges}
\end{figure}

Next we fixed $\epsilon$, i.e. the number of couples is $n^{1-\epsilon}$. Figure \ref{fig:epsilon} shows that the probability for finding a stable match with SoDA increases and is roughly concave in the number of singles. Observe that the rate of convergence is  different for various $\epsilon$'s.

\begin{figure}[htb]
\centering
{\includegraphics[width=10cm,height=6cm]{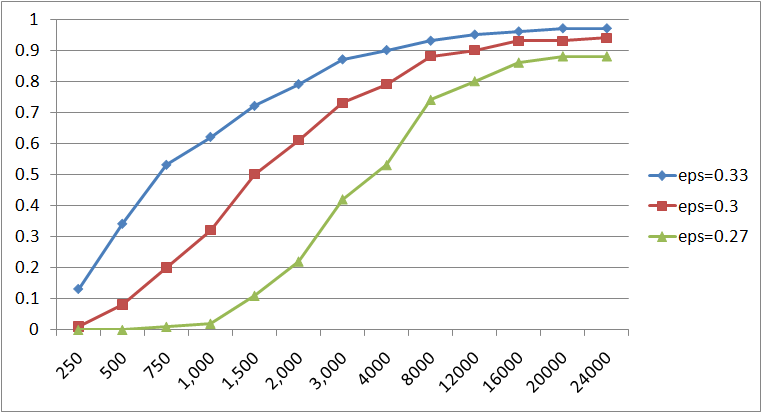}}
\caption{The success rate in finding a stable outcome given the number of singles (x-axis), where the number of couples is $n^{1-\epsilon}$ for three different $\epsilon$'s.}
\label{fig:epsilon}
\end{figure}

In the next simulation (see  Figure \ref{fig:regHist}) we fixed the number of singles and the number of couples to be 16,000 and 800 respectively as in the NMRP, and found the percentage of singles and couples that get their $k$-th most preferred choice. We assumed that there is no fitness, i.e. preference distributions of both doctors and hospitals are uniform.

\begin{figure}[htb]
\centering
{\includegraphics[width=10cm,height=6cm]{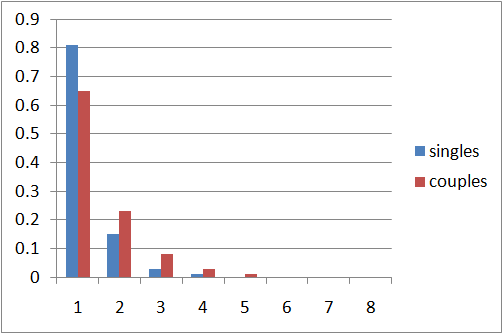}}
\caption{The histogram shows the percentage of singles and couples that got their k-th favorite choice for each
$k=1,\ldots,8$.}
\label{fig:regHist}
\end{figure}

\begin{figure}[htb]
\centering
{\includegraphics[width=10cm,height=6cm]{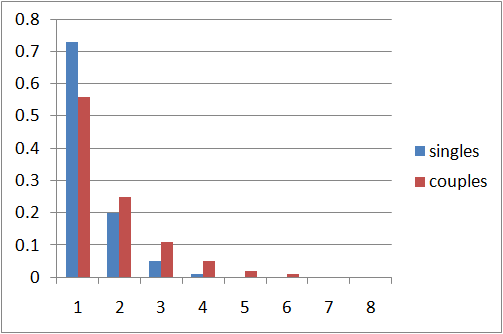}}
\caption{The histogram shows the percentage of singles and couples that got their k-th most preferred choice for each
$k=1,\ldots,8$. Hospitals have a fitness score.}
\label{fig:hospFitHist}
\end{figure}

In Figure \ref{fig:hospFitHist} we provide the same histogram but adding fitness to hospitals; each hospital has been assigned a score uniformly at random from the interval $[0.2,1]$. To decide the next preference of a doctor, she randomizes uniformly a hospital $h$ and a number from $[0.2,1]$, and if $h$'s score is below the number, the doctor resamples such a pair.

\section{`Almost' Linear is Necessary}

In Section \ref{sec:stability} we showed that the SoDA algorithm  finds a stable matching with probability  approaching  1 as $n$ tends to infinity assuming the number of couples is growing at a rate of $n^{1-\epsilon}$ (for any  $0<\epsilon<1$). In Section \ref{sec:simul} we saw that when the number of couples is a constant fraction of the total capacity, there is a constant probability of failure.
One might suggest that the SoDA algorithm does not search through enough permutations and if it fails there might still exist a stable matching.
We show  that not only SoDA will fail with constant probability but also any other algorithm, i.e.  with constant probability a stable matching does not exist.
For simplicity we will consider only uniformly distributed preferences and a capacity of 1 for each hospital.

\begin{theorem}
  Consider a random matching market with $n$ couples and  $n$ singles, $\lambda n$ hospitals for sufficiently large  $\lambda$  each of capacity 1, and preferences   distributed uniformly.  Then with some probability $delta > 0$ not depending on $n$, no stable matching exists.
  \footnote{The result is true also for $\alpha n$ couples for any constant $\alpha>0$.}
  \label{thm:linear1}
\end{theorem}
\begin{proof}
Consider the following event $E$: there exist a couple $c=(m_{c},f_{c})\in C$, a single doctor $s\in S$ and two hospitals $h_1\neq h_2$ so that the most preferred pair of hospitals by $c$ is $(h_1,h_2)$ and the following properties hold:
\begin{enumerate}
\item[(i)]
$h_2\succ_{s}h_1\succ_{s}h$ for any $h\notin\{h_1,h_2\}$.

\item[(ii)] $s\succ_{h_1} m_{c}$.

\item[(iii)] $f_{c}\succ_{h_2} s$.
\end{enumerate}
Observe that if only the couple $c$ and the single doctor $s$ existed  no stable matching would exists.

The proof will follow by first bounding (from below) the probability of the event $E$ and then bounding (from above) the event that some other doctor except  those in the event $E$ ever obtains either $h_1$ or $h_2$ in any  stable matching.

Fix a couple $c\in C$ and a single $s$ and let $(h_1,h_2)$ be the pair of hospitals most preferred by $c$. The probability that $h_1\neq h_2$,  and properties (ii) and (iii) hold is $\delta>\frac{1}{2}\cdot\frac{1}{2^2}$.  The probability that $h_1\neq h_2$ and properties  (i)-(iii) hold is $\Omega\left(\delta\frac{1}{(\lambda n)^2}\right)=\Omega\left(\frac{1}{n^2}\right)$. Therefore, since there are $n$ couples the probability that for a given single $s$  there exists a couple $c$ such that  $h_1\neq h_2$ and properties  (i)-(iii) hold is $\Omega\left(\frac{1}{n}\right)$. Therefore since there are $n$ singles, the probability that there exist a single  $s$ such that the event $E$ holds is some constant $\gamma>0$.

Suppose the event $E$ occurs with the couple $c'$ and doctor $s'$ and let $D'=D\setminus\{f_{c'},m_{c'},s'\}$.
Consider the following application/rejection algorithm in which doctors are assigned to $l>0$ positions (rather than 1):

{\it $l$-Pessimistic DA: At each step $t=1,2\ldots,$ either a single doctor $s\in S$ or a couple $c\in C$  that has less than $l$ temporary \emph{assignments} are chosen at random  and  apply to the most preferred hospital or pair of hospitals on their list respectively that they haven't applied so far. Each hospital assigns a doctor $d$ if and only if no other doctor is currently assigned to $h$ and no other doctor applied at this step to $h$. If some doctor $d$ applies to $h$ and some  doctor $d'$ (could be that $d'=d$) is temporarily assigned to $h$, $h$ rejects both $d$ and $d'$.}\footnote{As usual if a member of a couple is rejected from some hospital, its other member is also rejected.}

We will first show that the probability that any doctor but $f_{c'},m_{c'}$ and $s'$ ever applies  to $h_1$ or $h_2$ in the 3-Pessimistic DA process is bounded from above by a small constant. We will show a stronger lemma:

\begin{lemma} With constant probability no more than $\alpha n$ hospitals are visited in this process  for some $\alpha<\lambda$ in the 3-Pessimistic DA.
\label{lemmaPDA}
\end{lemma}
\begin{proof}
Let $L=\{0,1,2,3\}$.  For every $q\in L$ we say that a doctor is $q$-{\it settled} if it is temporarily assigned to exactly $q$ positions and we say that a hospital $h$ is {\it visited} if some doctor applied to it during the $3$-Pessimistic DA process.

For every $t=0,1,2\ldots,$ and every $q\in L$ denote by $A_t^q$ the number of $q$-settled doctors at step $t$, by $V_t$ the number of visited hospitals up to step $t$, where $A_0^0=3n$, and $A_0^1=A_0^2=A_0^3=V_0=0$.
Let $Y_t=V_t+15 A_t^0+10 A_t^1+5 A_t^2$ and consider the  process $X_t=Y_t+t$ for every $t=0,\ldots,\min(J,K)$, where $K$ is the first step in which $V_K=\frac{\lambda n}{10}$ and $J$ is the first step in which $A_J^0=A_J^1=A_J^2=0$ (i.e. $A_J^3=3n$).

\noindent{\bf Claim:}   $X_1,X_2\ldots,$ is a super-martingale, that is for every $t>0$, $E[X_{t+1}|X_1,\ldots,X_t]\leq X_t$.

\noindent{\bf Proof:}  Suppose a couple $c$ is chosen at step $t$ and has $q\in L\setminus\{3\}$ temporary assignments.  If it applies to two unvisited hospitals then $A_{t+1}^{q+1}=A_t^{q+1}+2$ and $A_{t+1}^q=A_t^q-2$ and $A_{t+1}^{q'}=A_t^{q'}$ for $q'\in L\setminus\{q,q+1\}$. Thus the contribution of the couple to $Y_t$ drops by $10$.
 If $c$ applies to an unvisited hospital and one  visited hospital then for every  $q\in L$, $A_{t+1}^q\leq A_{t}^q+2$ since at most one other couple lost a temporary assignment.  If it applies to two  visited hospitals then for every $q\in L$, $A_{t+1}^q\leq A_{t}^q+4$ since at most 2 additional couples lose a temporary assignment. For singles similar bounds can be used. For each $q=0,1,2$ let $Q_t^q$ be the event that at the beginning of step $t$ a couple with $q$ temporary assignments is chosen, and by  $W_t^q$ the event that a single with $q$ temporary assignments is chosen. Therefore for every $q\in L\setminus\{3\}$
\begin{multline*}
E[X_{t+1}|X_1,\ldots,X_t,Q_{t+1}^q]=E[X_{t+1}|X_t, Q_{t+1}^q] \leq \frac{(\lambda n - V_t)^2}{(\lambda n)^2}\left(V_t + 2 + 15 A_t^0+10 A_t^1+5 A_t^2 -10 \right) + \\ 2\cdot \frac{(\lambda n - V_t)V_t}{(\lambda n)^2}\left(V_t+1+15 A_t^0+10 A_t^1+5 A_t^2 +10\right) + \\ \frac{V_t^2}{(\lambda n)^2}\left(V_t+15 A_t^0+10 A_t^1+5 A_t^2+20\right)+t+1\leq V_t+ 15 A_t^0+10 A_t^1+5 A_t^2+t,
\end{multline*}
where the last inequality holds for any $V_t\leq \frac{\lambda n}{10}$. Similarly,
\begin{multline*}
E[X_{t+1}|X_1,\ldots,X_t,W_{t+1}^q]=E[X_{t+1}|X_t, W_{t+1}]\leq \frac{(\lambda n - V_t)}{\lambda n}\left(V_t +1 + 15 A_t^0+10 A_t^1+5 A_t^2-5\right)  + \\ \frac{V_t}{\lambda n}\left(V_t+15A_t^0+10A_t^1+5A_t^2+10\right)+t+1\leq V_t+ 15 A_t^0+10 A_t^1+5 A_t^2+t.
\end{multline*}
Therefore since either a couple or a single are chosen at each step, we obtain that $E[X_{t+1}|X_t]\leq V_t+ 15 A_t^0+10 A_t^1+5 A_t^2+t$.
$\Box$

As argued in the claim $|X_{t+1}-X_{t}|<22$ for every $t>1$.  Therefore by Azuma-Hoeffding's inequality we obtain that for any 
$T\geq 1$

\[\Pr\left(V_T-V_0\geq \frac{\lambda n}{10}\right)\leq \Pr\left(X_T-X_0\geq \frac{\lambda n}{10}-45n+T\right)\leq e^{-\frac{(\frac{\lambda n}{10}-45n+ T)^2}{968T }}<1-\beta,\]
for some constant $\beta>0$ and a sufficiently large $\lambda$, i.e. with constant probability the process will never reach $\frac{\lambda n}{10}$ visited hospitals.
\end{proof}

Lemma \ref{lemmaPDA} provides that in the 3-Pessimistic DA process described above, the number of hospitals visited is with constant probability only a  fraction of the total hospitals will be  visited, also implying that the doctors in the process (all but $c'$ and $s'$) will never visit $h_1$ and $h_2$.

Lemma \ref{lemmaPDA}  also implies that with constant probability the 3-Pessimistic DA  terminates and each {\it player} i, single  or couple, obtains 3 different temporary assignments, $p_i^1$,$p_i^2$ and $p_i^3$  (thus if $i$ is couple, $p_i^j$ is a pair of hospitals) and observe that $p_i^1\succ_i p_i^2\succ_i p_i^3$.

To finish the proof we argue that in every stable matching, no agent $i$ will be  assigned to a pair of hospitals less preferred to $p_i^3$.
Call  a player $i$ (a single or a couple) which gets a hospital less preferred to $p_i^3$ {\it poor}, and let $K$ be the set of poor player. Suppose that $|K|=k>0$.
For a player $i$ to be poor, at least one hospital in each  $p_i^1$,$p_i^2$ and $p_i^3$ should be taken (if $i$ is a single then all $p_i^j$ are single hospitals and all should be taken).  Since for each two players $j,l$, $\{p_j^1,p_j^2,p_j^3\}\cap\{p_l^1,p_l^2,p_l^3\}=\emptyset$ there are at least $3k$ hospitals which need to be assigned.  These hospitals cannot be assigned to
players that are not poor (since they get better choices for themselves). Since there are only $k$ poor players, with a total of up to $2k$ doctors, they cannot be assigned to all $3k$   hospitals -- a contradiction.
\end{proof}

\noindent{\it Remark:} If one assumes that doctors have constant length preference lists,  the proof is significantly simpler; indeed one can show directly  that the probability that  $h_1$ and $h_2$ are not acceptable for any doctor is constant.


\section{Conclusion}


We showed using the  SoDA algorithm that if the number of couples grows at a rate of $|C^n|= n^{1-\epsilon}$, then there exists a stable matching with probability approaching $1$.  One can argue that ``in real life'' the number of couples is indeed a linear fraction of the number of doctors, and the rate $|C^n| = n^{1- \epsilon}$ does not make sense. However, our correctness proof is only a lower bound on the performance of the algorithm, and it may
perform much better in practice.
Moreover, note that if $\epsilon$ were equal to  $O(1/\log n)$, then the number of couples was a linear fraction of the number of singles. In face, our proof shows that the random market has a stable matching with probability at least $1 - (\log n)^{O(1/\epsilon)} / n^{\Omega(\epsilon)}$,  which converges to 1 even if $\epsilon = \Omega(\log \log n/\sqrt{\log n})$, and not just when $\epsilon$ is constant.

This means that we proved that the algorithm finds a stable outcome with probability approaching 1 even when the number of couples grows like $n / 2^{\sqrt{\log n} \cdot\log\log n}$. Such growth is close to linear. Empirically it is indeed hard to distinguish between such subpolynomial factors and constant factors when there are $n = 16,000$ doctors.

A few open problems that follow from this work are the following. In Theorem \ref{thm:linear1} and its proof we used a  large excess number of hospitals to obtain the negative result. To some extent we do not expect that fewer hospitals will improve the chances of obtaining a stable matching. Figure \ref{fig:prectanges}  suggests when there are $\alpha n$ couples, the probability for the SoDA algorithm to find a stable matching decreases with $\alpha$. We conjecture that this is true in general, i.e. the probability that there exist a stable matching (not necessarily found by SoDA) is decreasing with $\alpha$.

\bibliographystyle{plainnat}
\bibliography{couples}

\begin{thebibliography}{16}
\providecommand{\natexlab}[1]{#1}
\providecommand{\url}[1]{\texttt{#1}}
\expandafter\ifx\csname urlstyle\endcsname\relax
  \providecommand{\doi}[1]{doi: #1}\else
  \providecommand{\doi}{doi: \begingroup \urlstyle{rm}\Url}\fi

\bibitem[Gale and Shapley(1962)]{GaleShapley}
D.~Gale and L.~L. Shapley.
\newblock {College Admissions and the Stability of Marriage}.
\newblock \emph{American Mathematical Monthly}, 69:\penalty0 9--15, 1962.

\bibitem[Gul and Stacchetti(1999)]{Gul99}
F.~Gul and E.~Stacchetti.
\newblock {Walrasian Equilibrium with Gross Substitutes}.
\newblock \emph{Journal of Economic Theory}, 87:\penalty0 95--124, 1999.

\bibitem[Hatfield and Kominers(2009)]{Kominers1}
J.~W. Hatfield and S.~D. Kominers.
\newblock {Many-to-Many Matching with Contracts}.
\newblock Working paper, 2009.

\bibitem[Immorlica and Mahdian(2005)]{NicoleMahdian}
N.~Immorlica and M.~Mahdian.
\newblock {Marriage, Honesty, and Stability}.
\newblock In \emph{Proc. of the sixteenth annual ACM-SIAM symposium on Discrete
  algorithms}, pages 53--62, 2005.

\bibitem[Klaus and Klijn(2005)]{KlausCouples1}
B.~Klaus and F.~Klijn.
\newblock {Stable Matchings and Preferences of Couples}.
\newblock \emph{Journal of Economic Theory}, 121:\penalty0 75--106, 2005.

\bibitem[Klaus et~al.(2009)Klaus, Klijn, and Nakamura]{KlausCouples2}
B.~Klaus, F.~Klijn, and T.~Nakamura.
\newblock {Corrigendum: Stable Matchings and Preferences of Couples}.
\newblock \emph{Journal of Economic Theory}, 144:\penalty0 2227--2233, 2009.

\bibitem[Kojima and Pathak(2009)]{KojimaParag}
F.~Kojima and P.~A. Pathak.
\newblock {Incentives and Stability in Large Two-Sided Matching Markets}.
\newblock \emph{American Economic Review}, 99:\penalty0 608--627, 2009.

\bibitem[Kojima et~al.(2010)Kojima, Pathak, and Roth]{RothCouples}
F.~Kojima, P.~Pathak, and A.~E. Roth.
\newblock {Matching with Couples: Stability and Incentives in Large Markets}.
\newblock Working paper, 2010.

\bibitem[Lahaie and Parkes(2009)]{ParkesLahaie}
S.~Lahaie and D.~C. Parkes.
\newblock {Fair Package Assignment}.
\newblock In \emph{Proceedings of the first Conference on Auctions, Market
  Mechasnisms and their Applications}, 2009.

\bibitem[Milgrom(2004)]{Milgrom04}
P.~R. Milgrom.
\newblock \emph{{Putting Auction Theory to Work}}.
\newblock Cambridge: Cambridge University Press, 2004.

\bibitem[Ning and Yang(2006)]{SunYang}
S.~Ning and Z.~Yang.
\newblock {Equilibria and Indivisibilities: Gross Substitutes and Complements}.
\newblock \emph{Econometrica}, 74:\penalty0 1385--1402, 2006.

\bibitem[Pycia(2010)]{Pycia1}
M.~Pycia.
\newblock {Stability and Preference Alignment in Matching and Coalition
  Formation}.
\newblock Working paper, UCLA, 2010.

\bibitem[Ronn(1990)]{Ronn}
E.~Ronn.
\newblock {NP-complete stable matching problems}.
\newblock \emph{J. Algorithms}, 11:\penalty0 285–--304, 1990.

\bibitem[Roth(1984)]{Roth1984}
A.~E. Roth.
\newblock {The Evolution of the Labor Market for Medical Interns and Residents:
  A Case Study in Game Theory}.
\newblock \emph{Journal of Political Economy}, 92:\penalty0 991--1016, 1984.

\bibitem[Roth(2009)]{RothCouplesHistory}
A.~E. Roth.
\newblock {The Origins, History, and Design of the Resident Match}.
\newblock \emph{The Origins, History, and Design of the Resident Match,"
  Journal of the American Medical Association}, 289:\penalty0 909--912, 2009.

\bibitem[Roth and Peranson(1999)]{RothPeranson}
A.~E. Roth and E.~Peranson.
\newblock {The Redesign of the Matching Market for American Physicians: Some
  Engineering Aspects of Economic Design}.
\newblock \emph{American Economic Review}, 89:\penalty0 748--780, 1999.

\end{thebibliography}


\end{document}